\documentclass[11pt]{article}
\usepackage[utf8]{inputenc}

\usepackage{fullpage}

\usepackage{xcolor}
\usepackage{amsmath}
\usepackage{amssymb}
\usepackage{amsthm}
\usepackage{algorithm}
\usepackage{algorithmic}
\usepackage[algo2e,ruled,noend]{algorithm2e}
\usepackage{hyperref}
\usepackage{cleveref}
\usepackage{graphicx,nicefrac}
\usepackage{subcaption}
\usepackage[font={small}]{caption}
\usepackage{thm-restate}
\usepackage[normalem]{ulem}
\usepackage{float}
\usepackage{balance}

\usepackage[shortlabels]{enumitem}
\usepackage{xcolor}
\usepackage[framemethod=TikZ]{mdframed}

\newcommand{\nc}{\newcommand}
\newcommand{\DMO}{\DeclareMathOperator}
\DeclareMathAlphabet\mathbfcal{OMS}{cmsy}{b}{n}

\allowdisplaybreaks

\nc{\MS}{\mathcal{S}}
\nc{\MP}{\mathcal{P}}
\nc{\MR}{\mathcal{R}}
\nc{\cM}{\mathcal{M}}
\nc{\cS}{\mathcal{S}}
\nc{\cI}{\mathcal{I}}
\nc{\cA}{\mathcal{A}}
\nc{\tcA}{\tilde{\cA}}

\nc{\MZ}{\mathcal{Z}}

\DMO{\Binom}{Binom}
\newcommand{\E}{\mathbb{E}}
\DMO{\Var}{Var}

\newcommand{\bX}{\mathbf{X}}
\nc{\tbx}{\tilde{\bx}}
\nc{\tbX}{\tilde{\bX}}
\nc{\tZ}{\tilde{Z}}
\nc{\tz}{\tilde{z}}
\newcommand{\bU}{\mathbf{U}}
\nc{\tbU}{\tilde{\bU}}
\newcommand{\bT}{\mathbf{T}}
\nc{\tbT}{\tilde{\bT}}
\newcommand{\bD}{\mathbf{D}}
\nc{\tbD}{\tilde{\bD}}

\newcommand{\bx}{\mathbf{x}}

\newcommand{\N}{\mathbb{N}}

\nc{\BN}{\mathbb{N}}

\nc{\BZ}{\mathbb{Z}}

\newcommand{\bone}{\mathbf{1}}

\newcommand{\bA}{\mathbf{A}}
\nc{\tbA}{\tilde{\bA}}

\newcommand{\cD}{\mathcal{D}}

\DeclareMathOperator{\supp}{supp}

\newcommand{\eps}{\epsilon}
\newcommand{\tO}{\widetilde{O}}

\newcommand{\rad}{\mathrm{rad}}

\newcommand{\tOmega}{\tilde{\Omega}}

\nc{\cost}{\mathrm{cost}}
\nc{\fcost}{\mathrm{f\text{-}cost}}
\nc{\ccost}{\mathrm{c\text{-}cost}}
\nc{\bN}{\mathbf{N}}
\nc{\cT}{\mathcal{T}}
\nc{\minset}{\mathrm{min\text{-}set}}
\nc{\opt}{\mathrm{OPT}}
\nc{\tf}{\tilde{f}}
\nc{\rsg}{\leftrightsquigarrow}
\nc{\bff}{\mathbf{f}}

\newtheorem{theorem}{Theorem}

\newtheorem{lemma}[theorem]{Lemma}

\newtheorem{definition}[theorem]{Definition}

\title{Improved Lower Bound for Differentially Private Facility Location}
\author{{Pasin Manurangsi}\\ Google Research, Bangkok, Thailand\\ \texttt{pasin@google.com}}
\date{\today}

\begin{document}

\maketitle

\begin{abstract}
We consider the differentially private (DP) facility location problem in the so called \emph{super-set output} setting proposed by Gupta et al.~\cite{GLMRT10}. The current best known expected approximation ratio for an $\eps$-DP algorithm is $O\left(\frac{\log n}{\sqrt{\eps}}\right)$ due to Cohen-Addad et al.~\cite{CEFG22} where $n$ denote the size of the metric space, meanwhile the best known lower bound is $\Omega(1/\sqrt{\eps})$~\cite{EGLW19}.

In this short note, we give a lower bound of $\tOmega\left(\min\left\{\log n, \sqrt{\frac{\log n}{\eps}}\right\}\right)$ on the expected approximation ratio of any $\eps$-DP algorithm, which is the first evidence that the approximation ratio has to grow with the size of the metric space.
\end{abstract}

\section{Introduction}

In the (metric) facility location (FL) problem, we are given a metric space $\cM = (V, d)$, facility costs $\bff = (f_v)_{v \in V}$ and input clients $\bX = (x_1, \dots, x_m) \in V^m$. The goal is to output a set $S \subseteq V$ of facilities to open to minimize its cost, which is defined as
\begin{align*}
\cost_{\cM, \bff, \bX}(S) = \sum_{s \in S} f_s + \sum_{i \in [m]} \min_{u \in S} d(x_i, u).
\end{align*}
The first term $\sum_{u \in S} f_s$ is often referred to as the \emph{facility cost} while the term $\sum_{i \in [m]} \min_{u \in S} d(x_i, u)$ is the \emph{connection cost}. Furthermore, we write $\opt_{\cM, \bff, \bX}$ to denote the minimum cost among all possible sets $S$, i.e. $\opt_{\cM, \bff, \bX} := \min_{S \subseteq V} \cost_{\cM, \bff, \bX}(S)$.

This is a classic problem in combinatorial optimization which has seen significant amount of progress over the years~\cite{Hochbaum82,ShmoysTA97,JainV99,CharikarG99,GuhaK99,KorupoluPR00,JainMS02,JainMMSV03,ChudakS03,MahdianYZ06,ByrkaA10,Li13}. Currently, the best known approximation ratio is 1.488~\cite{Li13} while it is known to be hard to approximate to within a factor of 1.463~\cite{GuhaK99}.

In recent years, privacy concern has led to numerous works on algorithms that preserve the users' privacy. Here we will use the notion of differential privacy (DP). To define DP, we say that two input vectors $\bX, \bX'$ are adjacent iff they differ on a single coordinate:

\begin{definition}[Adjacent Datasets]
Two datasets $\bX = (x_1, \dots, x_m) \in V^m$ and $\bX' = (x'_1, \dots, x'_m) \in V^m$ are \emph{adjacent} if there exists $i^* \in [m]$ such that $x_i = x'_i$ for all $i \in [m] \setminus \{i^*\}$.
\end{definition}

At a high level, an algorithm is DP if, when running it on two adjacent datasets, the output distributions are similar. This is formalized below.

\begin{definition}[Differential Privacy~\cite{DworkMNS06}] \label{def:dp}
For $\eps \geq 0$, an algorithm $A$ is said to be $\eps$-differentially private (or $\eps$-DP) if, for any adjacent inputs $\bX, \bX'$ and any output $o$, we have $\Pr[A(\bX) = o] \leq e^\eps \cdot \Pr[A(\bX') = o]$.
\end{definition}

DP provides a rigorous privacy guarantee for the users' data and enjoys several nice properties. For more background on the topic, we refer the readers to~\cite{DworkR14}.

As alluded to earlier, many combinatorial optimization problems have been studied under the additional DP restriction. Here, the question is: for a given $\eps$, what is the best approximation ratio\footnote{Sometimes additive approximation is also allowed in addition to multiplicative approximation.} an $\eps$-DP algorithm can achieve? Indeed, this question was studied for the facility location problem by Gupta et al.~\cite{GLMRT10}. Unfortunately, they showed that any $O(1)$-DP algorithm must have approximation ratio at least $\Omega(\sqrt{n})$ where $n = |V|$.

Due to the aforementioned barrier, Gupta et al.  proposed to study a slightly different setting called the \emph{super-set output} setting. Roughly speaking, this means that the algorithm will declare for every vertex in $V$ which facility should serve it. Then, the facility cost is computed with respect to only the facilities that serve at least one client in the input $\bX$. This is formalized below.

\begin{definition}[Super-set Output Facility Location (SOFL)]
In the \emph{super-set output facility location (SOFL)} problem, we are again given a metric space $\cM = (V, d)$, facility costs $\bff = (f_v)_{v \in V}$ and input clients $\bX = (x_1, \dots, x_m) \in V^m$. The goal is to output a mapping $\psi: V \to V$. The set of open facilities for $\psi$ (w.r.t. $\bX$) is defined as $S_{\psi, \bX} := \psi(\{x_1, \dots, x_m\})$. The cost of the solution $\psi$ is 
\begin{align*}
\cost_{\cM, \bff, \bX}(\psi) := \sum_{s \in S_{\psi, \bX}} f_s + \sum_{i \in [m]} d(x_i, \psi(x_i)).
\end{align*}

We say that a randomized algorithm $A$ has expected approximation ratio $\alpha$ iff, for all $\cM, \bff, \bX$,
\begin{align*}
\E_{\psi \sim A(\bX; \cM, \bff)}[\cost_{\cM, \bff, \bX}(\psi)] \leq \alpha \cdot \opt_{\cM, \bff, \bX}.
\end{align*}
\end{definition}

Note that, although the values of $\psi(v)$ for $v \notin \{x_1, \dots, x_m\}$ does not affect  $\cost_{\cM, \bff, \bX}(\psi)$, they may still need to be carefully constructed to respect the privacy constraint (\Cref{def:dp}).

Gupta et al. \cite{GLMRT10} gave an $\eps$-DP algorithm for SOFL with expected approximation ratio $\tO\left(\frac{(\log n \log \Delta)^2}{\eps}\right)$ where $\Delta$ is the diameter of $\cM$. Later, Esencayi et al.~\cite{EGLW19} and Cohen{-}Addad et al.~\cite{CEFG22} improved\footnote{Initially, the approximation ratio claimed in~\cite{EGLW19} was $O\left(\frac{\log n}{\eps}\right)$ but there is an issue in the analysis. Cohen{-}Addad et al.~\cite{CEFG22} both fixed this issue and improved the ratio to $O\left(\frac{\log n}{\sqrt{\eps}}\right)$.} the approximation ratio to $O\left(\frac{\log n}{\sqrt{\eps}}\right)$. In all these works, the first step is always to embed the general metric space $\cM$ to a tree metric~\cite{FRT04}; however, it is well known that such embedding must have (expected) distortion at least $\Omega(\log n)$~\cite{Bartal98}. Thus, such an approach cannot achieve an approximation ratio better than $\Omega(\log n)$. Meanwhile, the only lower bound for the problem is $O(1/\sqrt{\eps})$~\cite{EGLW19}. This leaves us with an intriguing question: \emph{Does the approximation ratio for DP SOFL in the super-set output setting have to grow with $n$?}

\subsection{Our Contribution}

We answer this question positively by showing the following lower bound:
\begin{theorem} \label{thm:lb-general}
For any $0 < \eps \leq O(1)$ and any sufficiently large $n$, there exists a metric space $\cM$ (with $n$ points) such that any $\eps$-DP algorithm for SOFL on $\cM$ in the super-set output setting must incur an expected approximation ratio of at least $\Omega\left(\min\left\{\frac{\log n}{\log \log n}, \sqrt{\frac{\log n}{\eps}}\right\}\right)$.
\end{theorem}

\paragraph{Proof Overview.} Our proof is based on the so-called ``packing'' framework~\cite{HardtT10}: We construct many small datasets such that any output $\psi$ can be a good approximate solution to only a small fraction of these datasets. By DP property and with appropriate parameters, this will allow us to conclude that no $\eps$-DP algorithm achieves a good approximation ratio.

The datasets we use are quite simple. We pick $\cM$ to be a shortest path metric of some graph $G$. Then, the datasets we consider are all the $m$ vertices that share a neighbor. To see why these are hard datasets, let us imagine for the moment that the graph $G$ is some random graph. For $\psi$ to have a low connection cost, each vertex has to be mapped to a nearby vertex (in shortest path distance). However, we show below that this implies that most vertices in a random dataset (as described earlier) are mapped to different facilities. Thus, the number of facilities open are large, leading to a large facility cost. We formalize this argument below. In fact, we do not need $G$ to be fully random but rather we only require that its girth is large.

\section{Preliminaries}

For an undirected unweighted graph $G = (V, E)$, we write $d_G$ to denote its shortest path metric, i.e. $d_G(u, v)$ is equal to the shortest path distance between $u$ and $v$. We write $N_G(v)$ to denote the set of neighbors of $v$ in $G$. The \emph{girth} of $G$ is the length of its smallest cycle. For two vertices $u, v$ such that $d_G(u, v) < g/2$, there must be a unique shortest path between them and we write $P_{u \rsg v}$ to denote this path.

We will use a classic result of Erdos and Sachs on the existence of regular large-girth graphs\footnote{An English version of the statement and its proof can be found in~\cite[Appendix C]{EJ08}; the bound we state here is weaker than the ones stated there but is sufficient for our purpose.}. 

\begin{theorem}[{\cite{erdos1963regulare}}] \label{thm:large-girth-graph}
For any integers $g \geq 3, d \geq 2$ and any even integer $n \geq 4d^g$, there is an $n$-vertex $d$-regular graph with girth at least $g$. 
\end{theorem}

\section{Main Proof}

For $\psi: V \to V$, we write $\rad_G(\psi)$ to denote $\max_{v \in V} d_G(v, \psi(v))$. As mentioned in the proof overview, a key lemma is to show that $\psi$ maps most points in a random vertex's neighborhood to different facilities. This is formalized below.

\begin{lemma} \label{lem:neighbor-map-size-bound}
Let $G$ be any $n$-vertex $d$-regular graph with girth at least $g$.
For any $\psi: V \to V$ such that $\rad_G(\psi) < g / 2 - 1$, we have
\begin{align*}
\Pr_{v \in V}[|\psi(N_G(v))| < d - \sqrt{d}] < 1/\sqrt{d}.
\end{align*}
\end{lemma}

\begin{proof}
Fix $v \in V$. We have
\begin{align*}
d - |\psi(N_G(v))| &= \sum_{w \in V} \left(|\psi^{-1}(w) \cap N_G(v)| - \bone[\psi^{-1}(w) \cap N_G(v) \ne \emptyset]\right) \\
&= \sum_{w \in V} \max\{|\psi^{-1}(w) \cap N_G(v)| - 1, 0\}. 
\end{align*}
Consider any fixed $w, v \in V$ such that $\psi^{-1}(w) \cap N_G(v) \ne \emptyset$.
Note that, since $\rad_G(\psi) < g / 2 - 1$, any $u \in \psi^{-1}(w) \cap N_G(v)$ satisfies $d_G(w, u) < g/2 - 1$ and $d_G(u, v) = 1$. This also means that $d_G(w, v) \leq d_G(w, u) + d_G(u, v) < g/2$. Since the girth of the graph is $g$, it must be the case that (i) $d_G(w, u) = d_G(w, v) - 1$ and $u\in P_{w \rsg v}$ or (ii) $d_G(w, u) = d_G(w, v) + 1$ and $v \in P_{w \rsg u}$. There is just one vertex $u$ in case (i), i.e. the vertex right next to $v$ in the path $P_{w \rsg v}$. Thus, $|\psi^{-1}(w) \cap N_G(v)| - 1$ is at most the number of vertices in case (ii). Plugging this into the above, we get
\begin{align*}
d - |\psi(N_G(v))| \leq \sum_{w \in V} |\{u \in N_G(v) \mid \psi(u) = w \wedge d_G(w, u) = d_G(w, v) + 1, v \in P_{w \rsg u}\}|.
\end{align*}
Summing this up over all $v \in V$, we get
\begin{align*}
&\sum_{v \in V} \left(d - |\psi(N_G(v))|\right) \\
&\leq \sum_{v \in V} \sum_{w \in V} |\{u \in N_G(v) \mid \psi(u) = w \wedge d_G(w, u) = d_G(w, v) + 1, v \in P_{w \rsg u}\}| \\
&= \sum_{v \in V} \sum_{w \in V} \sum_{u \in V} \bone[\psi(u) = w]\bone[d_G(w, u) = d_G(w, v) + 1, v \in P_{w \rsg u}] \\
&= \sum_{v \in V} \sum_{u \in V} \bone[d_G(\psi(u), u) = d_G(\psi(u), v) + 1, v \in P_{\psi(u) \rsg u}] \\
&= \sum_{u \in V} |\{v \in V \mid d_G(\psi(u), u) = d_G(\psi(u), v) + 1, v \in P_{\psi(u) \rsg u}\}| \\
&\leq \sum_{u \in V} 1 \\
&= n,
\end{align*}
where the second inequality is because, for each fixed $u \in V$, there is just one $v \in V$ that satisfies $d_G(\psi(u), u) = d_G(\psi(u), v) + 1, v \in P_{\psi(u) \rsg u}$, i.e. the vertex right next to $u$ in the path $P_{\psi(u) \rsg u}$. 

Applying Markov's inequality to the above sum yields the desired result.
\end{proof}

Since we will not construct the dataset by taking the entire neighborhood of $v$ but rather by sampling random $m$ neighbors, the following lemma will be more convenient.

\begin{lemma} \label{lem:dataset-partition}
Let $G$ be any $n$-vertex $d$-regular graph with girth at least $g$ and any $m \in \N$. For any $\psi: V \to V$ such that $\rad_G(\psi) < g / 2 - 1$, we have
\begin{align*}
\Pr_{v \sim V, u_1, \dots, u_m \sim N_G(v)}[|\psi(\{u_1, \dots, u_m\})| < m] < m^2/\sqrt{d},
\end{align*}
where $u_1, \dots, u_m \sim N_G(v)$ denote $m$ i.i.d. sample drawn u.a.r. from $N_G(v)$.
\end{lemma}

\begin{proof}
We have
\begin{align*}
&\Pr_{v \sim V, u_1, \dots, u_m \sim N_G(v)}[|\psi(\{u_1, \dots, u_m\})| < m] \\
&\leq \Pr_{v \sim V}[|\psi(N_G(v))| < d - \sqrt{d}] + \Pr_{v \sim V, u_1, \dots, u_m \sim N_G(v)}[|\psi(\{u_1, \dots, u_m\})| < m \mid |\psi(N_G(v))| \geq d - \sqrt{d}] \\
&< \frac{1}{\sqrt{d}} + \Pr_{v \sim V, u_1, \dots, u_m \sim N_G(v)}[|\psi(\{u_1, \dots, u_m\})| < m \mid |\psi(N_G(v))| \geq d - \sqrt{d}],
\end{align*}
where the second inequality is from \Cref{lem:neighbor-map-size-bound}.

To bound the second term, let us fix any $v \in V$ such that $|\psi(N_G(v))| \geq d - \sqrt{d}$. Note that this implies that $\max_{w \in V} |\psi^{-1}(w) \cap N_G(v)| \leq \sqrt{d} + 1 \leq 2\sqrt{d}$. From this, we can derive
\begin{align*}
\Pr_{u_1, \dots, u_m \sim N_G(v)}[|\psi(\{u_1, \dots, u_m\})| < m]
&= \Pr_{u_1, \dots, u_m \sim N_G(v)}[\exists 1 \leq i < j \leq m, \psi(u_i) = \psi(u_j)] \\
&\leq \sum_{1 \leq i < j \leq m} \Pr_{u_i, u_j \sim N_G(v)}[\psi(u_i) = \psi(u_j)] \\
&= \sum_{1 \leq i < j \leq m} \frac{1}{d^2} \left(\sum_{w \in V} |\psi^{-1}(w) \cap N_G(v)|^2\right) \\
&\leq \sum_{1 \leq i < j \leq m} \frac{2\sqrt{d}}{d^2} \left(\sum_{w \in V} |\psi^{-1}(w) \cap N_G(v)|\right) \\
&= \sum_{1 \leq i < j \leq m} \frac{2}{\sqrt{d}} \\
&\leq \frac{m^2 - 1}{\sqrt{d}}
\end{align*}
where in the second inequality we use the fact that $\max_{w \in V} |\psi^{-1}(w) \cap N_G(v)| \leq 2\sqrt{d}$.

Combining the two preceding inequalities yield our claimed bound.
\end{proof}

We can now complete the proof of \Cref{thm:lb-general} by following the standard packing lemma paradigm~\cite{HardtT10} where the dataset is constructed as in the above lemma.

\begin{proof}[Proof of \Cref{thm:lb-general}]
We assume w.l.o.g. that $n$ is even. We will prove a lower bound of $\Omega\left(\sqrt{\frac{\log n}{\eps}}\right)$ on the expected approximation ratio for $\eps > \eps^* := \frac{C(\log \log n)^2}{\log n}$ where $C$ is a sufficiently large constant. Note that, if $\eps \leq \eps^*$, we can simply use this lower bound for $\eps^*$ (since any $\eps$-DP algorithm is also $\eps^*$-DP in this case) which yields $\Omega\left(\frac{\log n}{\log \log n}\right) = \Omega\left(\min\left\{\frac{\log n}{\log \log n}, \sqrt{\frac{\log n}{\eps}}\right\}\right)$ as desired. Thus, we may henceforth assume that $\eps > \eps^*$.

We select our parameters as follows:
\begin{itemize}
\item Target approximation ratio lower bound: $\gamma := 0.0001 \cdot \sqrt{\frac{\log n}{\eps}}$,
\item Graph girth: $g := \sqrt{\frac{\log n}{\eps}}$,
\item Graph degree: $d := \lfloor n^{1/g} / 4 \rfloor = \exp\left(\Theta\left(\sqrt{\eps \cdot \log n}\right)\right)$,
\item Facility cost: $f := g/2 - 1 = \Theta\left(\sqrt{\frac{\log n}{\eps}}\right)$,
\item Dataset size: $m := \lceil 0.1 \log(d) / \eps\rceil = \Theta\left(\sqrt{\frac{\log n}{\eps}}\right)$.
\end{itemize}

Consider any $\eps$-DP algorithm $\cA$ for DPFL. Let $G = (V, E)$ be an $n$-vertex $d$-regular graph with girth at least $g$; \Cref{thm:large-girth-graph} implies the existence of such a graph for any sufficiently large $n$. Let the metric space be $\cM = (V, d_G)$ and let $f_v = f$ for all $v \in V$. We may assume w.l.o.g. that any output $\psi$ of $\cA$ satisfies $\rad_G(\psi) < f$. Otherwise, if there exists $u \in V$ such that $d_G(\psi(u), u) \geq f$, then we may simply let $\psi(u) = u$; this does not increase the cost of the solution and, since this is a post-processing, does not impact the DP guarantee of the algorithm~\cite[Proposition 2.1]{DworkR14}.

Suppose for the sake of contradiction that $\cA$ achieves an expected approximation ratio of $\gamma$ on $\cM$. Let $\cD$ denote the distribution of $\bX = (x_1, \dots, x_m)$ corresponding to picking $v \in V$ at random and picking $m$ of its neighbors $x_1, \dots, x_m$ randomly with replacement. \Cref{lem:dataset-partition} implies that, for any $\psi: V \to V$ such that $\rad_G(\psi) < f$, we have
\begin{align} \label{eq:packing1}
\Pr_{\bX \sim \cD}[\cost_{\cM, \bff, \bX}(\psi) < m \cdot f] < \frac{m^2}{\sqrt{d}} \leq \frac{0.1}{d^{0.1}},
\end{align}
where the second inequality holds for any sufficiently large $n$ due to our choice of $m, d$ and from $\eps \geq \frac{C(\log \log n)^2}{\log n}$.

Meanwhile, for any $\bX \in \supp(\cD)$, if we open a single facility at $v$, the cost is $f + m$; this means that $\opt_{\cM, \bff, \bX} \leq f + m$. Notice also that  $2\gamma(f + m) < m \cdot f$ by our choice of parameters. Thus, by the approximation guarantee of $\cA$, we have
\begin{align*}
0.5 \leq \Pr_{\psi \gets \cA(\bX)}[\cost_{\cM, \bff, \bX}(\psi) < m \cdot f] & &\forall \bX \in \supp(\cD).
\end{align*}
Let $\bX^*$ be any dataset in $V^m$.
By $\eps$-DP guarantee of $\cA$, we have\footnote{See e.g. \cite[Theorem 2.2]{DworkR14}.}
\begin{align*}
0.5 \cdot \exp(-\eps \cdot m) \leq \Pr_{\psi \gets \cA(\bX^*)}[\cost_{\cM, \bff, \bX}(\psi) < m \cdot f]. & &\forall \bX \in \supp(\cD).
\end{align*}
Taking the expectation of the above over $\bX \sim \cD$, we have
\begin{align*}
0.5 \cdot \exp(-\eps \cdot m) 
&\leq \Pr_{\bX \sim \cD, \psi \gets \cA(\bX^*)}[\cost_{\cM, \bff, \bX}(\psi) < m \cdot f] \\
&\leq \Pr_{\psi \gets \cA(\bX^*)}\left[\Pr_{\bX \sim \cD}[\cost_{\cM, \bff, \bX}(\psi) < m \cdot f]\right] \\
&\overset{\eqref{eq:packing1}}{<} \frac{0.1}{d^{0.1}},
\end{align*}
which is a contradiction since, from our choice of $m$, the LHS is at least $0.2/d^{0.1}$.
\end{proof}

\section{Open Questions}

An open question is to close the gap between our lower bound and the upper bound from \cite{CEFG22}. Perhaps the most representative case here is when $\eps$ is a constant, e.g. $\eps = 1$. Since all known algorithms embed the metric into tree metrics, these algorithms cannot achieve better than $O(\log n)$ approximation ratio. Meanwhile, it is unclear how to push our lower bound beyond $\tOmega(\sqrt{\log n})$ here. 

Another open question is whether one can extend our lower bound to hold against approximate-DP algorithms. It is not hard to see that the above argument also for $(\eps, \delta)$-DP algorithm where $\delta \ll -\exp(\eps \cdot m)$. However, a typically interesting regime of $\delta$ is when $\delta = m^{-\Theta(1)}$.

\bibliographystyle{alpha}
\bibliography{ref}

\end{document}